\documentclass[12pt]{article}

\usepackage{amsxtra,amssymb,amsthm,amsmath,latexsym}
\usepackage{graphicx}
\usepackage{epsfig}
\usepackage{epstopdf}
\usepackage{booktabs}
\usepackage{multirow}
\usepackage{afterpage}
\usepackage{setspace}
\usepackage{microtype}
\usepackage{calc,fourier}
\usepackage{cite}
\usepackage[T1]{fontenc}
\usepackage[hidelinks]{hyperref}
\usepackage{float}
\usepackage[small]{caption}
\newtheorem{thm}{Theorem}[section]

\newtheorem{lem}[]{Lemma}

\numberwithin{equation}{section}

\newcommand{\RRR}{\mathbb{R}^3}

\newcommand{\bee}{\begin{equation*}}
\newcommand{\eee}{\end{equation*}}
\newcommand{\be}{\begin{equation}}
\newcommand{\ee}{\end{equation}}
\newcommand{\bal}{\begin{align}}
\newcommand{\eal}{\end{align}}
\newcommand{\ba}{\begin{array}}
\newcommand{\ea}{\end{array}}

\begin{document}
\title{Electromagnetic wave scattering by small perfectly conducting particles and applications}

\author{
A. G. Ramm$\dag$\footnotemark[1]
\\
\\
$\dag$Mathematics Department, Kansas State University,\\
Manhattan, KS 66506-2602, USA
}

\renewcommand{\thefootnote}{\fnsymbol{footnote}}
\footnotetext[1]{Email: ramm@math.ksu.edu}
\date{}
\maketitle

\begin{abstract} \noindent
A rigorous theory of electromagnetic (EM) wave scattering by one and many perfectly conducting small bodies of an arbitrary shape is developed. Equation for the effective field is derived in a medium in which many small particles are distributed.
A method is given to change the refraction coefficient of a given medium in a desired direction by embedding into this medium many small particles.



{\bf PACS:}
0230Rz, 4225Fx, 8105Zx
\end{abstract}

\section{Introduction} \label{sec1}
The electromagnetic (EM) wave scattering by a small perfectly conducting particle was studied in many papers. Rayleigh (1871) initiated this study and understood that the main term in the scattered field is given by a dipole radiation, \cite{Ra}.
The smallness of the particle $D$ is characterized by the inequality $ka \ll 1$, where $a=\frac{1}{2}\text{diam}D$ is the characteristic size of the particle, $k=\frac{\omega}{c}$ is the wave number, $\omega$ is the frequency, $c$ is the wave speed. The dipole radiation is generated by the induced dipole polarization of the small body. In \cite{Ra} there was no method given for calculating the induced dipole moment for a body of an arbitrary shape. This was done in \cite{R144, R476}, where formulas were derived that allowed one to calculate the polarizability tensor for a body of an arbitrary shape and, therefore, the induced dipole moment. In paper \cite{Mie} the EM wave scattering problem was solved for spheres by the method of separation of variables. This method cannot be used for bodies of arbitrary shapes. The first basic new result of the current paper is an analytic explicit formula for the EM field scattered by a small perfectly conducting particle of an arbitrary shape
 and the method used for the derivation of this formula.

The second basic new result is a numerical method for solving EM wave scattering by $M$ small perfectly conducting bodies. The number $M$ can be very large, $M=O(10^{12})$, and $M=M(a)\to \infty$ as $a \to 0$.

 More precisely, the distribution of the small bodies is described as follows. The number $\mathcal{N}(\Delta)$ of small bodies in an arbitrary open set $\Delta \subset \Omega$, where $\Omega$ is an arbitrary fixed domain in $\RRR$, is given by the formula
\be \label{eq1.1}
    \mathcal{N}(\Delta)=\frac{1}{a^3} \int_{\Delta} N(x)dx(1+o(1)), \quad a \to 0.
\ee
Here $N(x)\ge 0$ is a continuous function that can be chosen arbitrarily by the experimentalist. If $D_m, 1 \le m \le M$, are the small non-intersecting bodies distributed in $\Omega$ according to the law \eqref{eq1.1}, then one has
\be \label{eq1.2}
    M=M(a)=\frac{1}{a^3} \int_\Omega N(x) dx[1+o(1)]=O\left(\frac{1}{a^3}\right), \quad a \to 0.
\ee
By $x_m$ we denote an arbitrary fixed point inside $D_m$. Thus,
\be \label{eq1.3}
    \mathcal{N}(\Delta)=\sum_{x_m \in \Delta}1.
\ee
The EM wave scattering problem for one small perfectly conducting body is formulated as follows:
\begin{align}
    &\nabla \times E=i\omega \mu H \quad\text{in }D':=\RRR \setminus D, \label{eq1.4}\\
    &\nabla \times H=-i\omega\epsilon E \quad\text{in } D',\label{eq1.5}
\end{align}
where $\omega, \mu$ are constants in $D'$, $k=\omega\sqrt{\epsilon\mu}$, $\epsilon$ and $\mu$ are dielectric permitivity and magnetic permeability satisfying Maxwell's equations \eqref{eq1.4}-\eqref{eq1.5} in $\RRR$, and $v_E, v_H$ be the scattered fields, satisfying equations \eqref{eq1.4}-\eqref{eq1.5} and the radiation condition
\be \label{eq1.6}
    r\left(\frac{\partial v_E}{\partial r}-ikv_E\right)=o(1), \quad r:=|x|\to \infty.
\ee
The boundary condition on the surface $S$ of the small body $D$ is
\be \label{eq1.7}
    [N,[E,N]]=0 \quad \text{on }S,
\ee
which means that tangential component of $E$ vanishes on $S$. Here $N$ is the unit normal to $S$ directed into $D'$, $[A,B]$ stands for the vector product, $A\cdot B$ stands for the scalar product.

Suppose that
\be \label{eq1.8}
    E_0=\mathcal{E}e^{ik\alpha\cdot x}, \quad \mathcal{E}\cdot \alpha=0, \quad H_0:=\frac{\nabla \times E_0}{i\omega \mu},
\ee
is the plane incident wave, $\mathcal{E}=const, \alpha \in S^2$, $S^2$ is the unit sphere in $\RRR$. Then
\be \label{eq1.9}
    \nabla \cdot E_0=0, \quad \nabla \cdot H_0=0,
\ee
and $E_0, H_0$ satisfy equations \eqref{eq1.4}-\eqref{eq1.5} in $\RRR$.

The scattered field
\be \label{eq1.10}
    v_E=\frac{e^{ikr}}{r}A(\beta,\alpha,k)+o\left(\frac{1}{r}\right), \quad r=|x|\to \infty, \quad \beta:=\frac{x}{r}.
\ee
The coefficient $A(\beta,\alpha,k)$ is called the scattering amplitude. If $v_E$ is known, then
\be \label{eq1.11}
    v_H=\frac{\nabla\times v_E}{i\omega\mu}.
\ee
In section \ref{sec2} a formula for $A(\beta,\alpha,k)$ is derived.

In section \ref{sec3} a numerical method is developed for solving many-body EM wave scattering problem in the case of small perfectly conducting bodies of an arbitrary shape and a limiting equation is derived for the effective field in the medium as $a\to 0$.

In section \ref{sec4} applications of our theory to materials science are discussed. It is explained how to change the original refraction coefficient in the desired direction.

In section \ref{sec5} it is stated that basic results of this paper are formulated in theorems \ref{thm2.3}, \ref{thm2.4}, \ref{thm3.2}, \ref{thm4.1}.

The ideas we use in this paper are similar to the ideas developed in \cite{R635}. In \cite{BH} one finds a recent report about light scattering by small particles.

\section{ EM wave scattering by one small perfectly conducting body} \label{sec2}
Let us look for a solution to problem \eqref{eq1.4}-\eqref{eq1.8} of the form
\be \label{eq2.1}
    E=E_0+\nabla\times\int_S g(x,t)J(t)dt, \quad g(x,t):=\frac{e^ik|x-t|}{4\pi |x-t|}.
\ee
Here $J(t)$ is a tangential to $S$ field that should be found from the boundary condition \eqref{eq1.7}. Thus, we look for \be \label{eq2.2}
    v_E:=\nabla\times\int_S g(x,t)J(t)dt.
\ee
Equation \eqref{eq1.4} is satisfied if one takes
\be \label{eq2.3}
    v_H:= \frac{\nabla\times v_E}{i \omega\mu}.
\ee
Equation \eqref{eq1.5} is satisfied also:
\be \label{eq2.4}
    \nabla\times v_H = \frac{\nabla\times\nabla\times\nabla\times\int_S g J dt}{i \omega\mu} =\frac{-\nabla^2\nabla\times\int_S g J dt}{i \omega\mu} =\frac{k^2 v_E}{i \omega\mu} =-i\omega\epsilon v_E.
\ee
This is true for any $J$. Let us prove that {\em $J$ is uniquely determined by the boundary condition \eqref{eq1.7}, namely, that equation \eqref{eq1.7} has at most one solution.}

 To prove this, it is sufficient to prove that equation \eqref{eq1.7} is of Fredholm-type for $J$, and that the corresponding homogeneous equation has only the trivial solution.

Let us write equation \eqref{eq1.7} for $J$:
\be \label{eq2.5}
    [N,[E_0,N]]+[N,[\nabla\times\int_S g(x,t)J(t) dt|_{x\to s^-}, N]]=0,
\ee
where $x\to s^-$ denotes the limit as $x\to s$ from outside $D$ along the normal to $S$ at the point $s$. We will use the known formula (see, for example, \cite{R635}, p.86):
\be \label{eq2.6}
    \lim_{x\to s^-}[N,\nabla\times\int_S g(x,t)J(t)dt] =\frac{J(s)}{2}+ \int_S [N_s,[\nabla_s g(s,t), J(t)]]dt.
\ee
Take a vector product of $N$ with the left side of equation \eqref{eq2.5}, use \eqref{eq2.6} and get
\bee
    [N,E_0]+[N,\nabla\times\int_S g(x,t)J(t)dt]_{x\to s^-}=0,
\eee
or
\be \label{eq2.7}
    \frac{J(s)}{2}+TJ:=\frac{J(s)}{2}+\int_S [N_s, [\nabla_s g(s,t), J(t)]]dt=-[N,E_0].
\ee
Equation \eqref{eq2.7} is equivalent to equation \eqref{eq2.5}: taking the vector product of \eqref{eq2.7} with $N$ one gets \eqref{eq2.5}.

\begin{lem} \label{lem2.1}
    Equation \eqref{eq2.7} is of Fredholm-type in the space $C(S)$ of continuous tangent to $S$ fields.
\end{lem}
\begin{proof}
    It is sufficient to check that the operator $T$ in \eqref{eq2.7} is compact in $C(S)$ and that any solution $J$ to equation \eqref{eq2.7} is tangential to $S$, that is,
    \be \label{eq2.8}
        N_s \cdot J(s)=0, \quad \forall s \in S.
    \ee
    To prove \eqref{eq2.8}, scalar multiply \eqref{eq2.7} by $N_s$. Since $N_s \cdot [N_s,E_0]=0$ and $N_s\cdot[N_s,[\nabla_s g, J(t)]]=0$, the desired relation \eqref{eq2.8} follows.

    Compactness of $T$ follows from the formula
    \be \label{eq2.9}
        TJ=\int_S \left(\nabla_s g(s,t)N_s\cdot J(t)-J(t)\frac{\partial g(s,t)}{\partial N_s}\right)dt,
    \ee
     relation \eqref{eq2.8}, and the estimate
    \be \label{eq2.10}
        \left|\frac{\partial g(s,t)}{\partial N_s}\right|=O\left(\frac{1}{|s-t|}\right), \quad |s-t|\to 0,
    \ee
    known from the potential theory for the $C^2$ surfaces (see, for example, \cite{R635}, chapter 11).

    Lemma \ref{lem2.1} is proved.
\end{proof}
\begin{lem} \label{lem2.2}
    The homogeneous version of equation \eqref{eq2.7} implies $J=0$, provided that $a$ is sufficiently small.
\end{lem}
\begin{proof}
    It is sufficient to prove that
    \be \label{eq2.11}
        ||T||=O(a), \quad a\to  0.
    \ee
    Estimate \eqref{eq2.8} implies that
    \be \label{eq2.12}
        |N_s \cdot J(t)|=O(|s-t|)||J||.
    \ee
    Therefore
    \be \label{eq2.13}
        \max_{s \in S} \int_S |\nabla_s g(s,t)||N_s \cdot J(t)|dt \le \max_{s \in S} \int_S O\left(\frac{1}{|s-t|}\right)dt||J|| =O(a)||J||.
    \ee
    Furthermore
    \be \label{eq2.14}
        \max_{s \in S} \int_S \left|\frac{\partial g(s,t)}{\partial N_s}\right|dt =O(a).
    \ee
    Estimates \eqref{eq2.13} and \eqref{eq2.14} imply \eqref{eq2.11}.

    Lemma \ref{lem2.2} is proved.
\end{proof}
From these lemmas the following theorem follows:
\begin{thm} \label{thm2.3}
    Equation \eqref{eq2.7} has a solution in $C(S)$, this solution is unique and satisfies condition \eqref{eq2.8}.
\end{thm}
If $r:=|x|\to \infty, \frac{x}{r}=\beta$, and the origin is inside $D$, then formula \eqref{eq2.1} implies
\be \label{eq2.15}
    A(\beta, \alpha, k)=\frac{ik}{4\pi}[\beta, Q], \quad Q:=\int_S J(t)dt.
\ee
Let us derive a formula for $Q$. Integrate equation \eqref{eq2.7} over $S$ and keep the main terms as $a \to 0$. One has
\be \label{eq2.16}
    \frac{Q}{2}+\int_S ds\int_S \left(\nabla_s g(s,t)N_s\cdot J(t)-J(t)\frac{\partial g(s,t)}{\partial N_s}\right)dt= -\int_S [N,E_0]ds.
\ee
Clearly,
\be \label{eq2.17}
    -\int_S [N,E_0]ds=-\int_D \nabla \times E_0 dx=-(\nabla \times E_0)(x_1)|D|,
\ee
where $x_1 \in D$, and $|D|=c_D a^3,$ $c_D>0$ is a constant which depends on the geometry of $D$. Furthermore
\be \label{eq2.18}
    \int_S dt J(t) \int_S\left(-\frac{\partial g(s,t)}{\partial N_s}\right)ds=\frac{Q}{2}, \quad a \to 0,
\ee
where we have used the relation
\be \label{eq2.19}
    -\int_S\frac{\partial g(s,t)}{\partial N_s}ds=\frac{1}{2}, \quad a \to 0,
\ee
see, for example, \cite{R635}, p.8.

Finally,
\begin{align} \label{eq2.20}
    I:&=\int_S ds\int_S dt \nabla_s g(s,t)N_s\cdot J(t) \nonumber \\
      &=\int_S dt J(t)\cdot\left(\int_{|s-t|<\delta}N_s\cdot \nabla_s g(s,t) ds + \int_{|s-t|\ge \delta}N_s\cdot \nabla_s g(s,t) ds\right) \nonumber \\
      &:= \int_S dt J(t)(I_1+I_2).
\end{align}
One has $|J(t)|\cdot N(s) \le c|t-s|,$ and $ |\nabla_s g(s,t)| \le \frac{c}{|s-t|^2}$,
 where $c>0$ is an estimation constant. Thus,
\be \label{eq2.21}
    |I_1| \le c\int_{|s-t|\le \delta} \frac{ds}{|s-t|} \le ca \epsilon(\delta), \quad
    \lim_{\delta \to 0}\epsilon(\delta)= 0,
\ee
and
\be \label{eq2.22}
    |I_2| \le c\int_{|s-t|\ge \delta} \frac{ds}{\delta^2} \le C\frac{a^2}{\delta^2}.
\ee
Therefore, choosing, for example,  $\delta=(a)^{1/2}$ and $a$ sufficiently small so that $\epsilon (\delta)$
is sufficiently small, one gets
\be \label{eq2.23}
    |I|\le cQ\left(a+ \epsilon((a)^{1/2})a\right)=o(1)Q, \qquad a\to 0.
\ee
From \eqref{eq2.16}-\eqref{eq2.18} and \eqref{eq2.23} one obtains
\be \label{eq2.24}
    Q=-(\nabla\times E_0)(x_1)|D|, \qquad a \to 0,
\ee
where $x_1 \in D$ is an arbitrary point.

From \eqref{eq2.24} and \eqref{eq2.15} it follows that
\be \label{eq2.25}
    A(\beta,\alpha,k)=-\frac{ik}{4\pi}[\beta,(\nabla\times E_0)(x_1)]|D|, \quad |D|=c_D a^3, \quad a \to 0,
\ee
and
\be \label{eq2.26}
    (\nabla \times E_0)(x_1)=ik[\alpha, \mathcal{E}]e^{ik\alpha\cdot x_1}.
\ee
Since $ka \ll 1$ and $x_1 \le a$, one may write
\be \label{eq2.26a}
    (\nabla\times E_0)(x_1)=ik[\alpha, \mathcal{E}],
\ee
and
\be \label{eq2.27}
    A(\beta,\alpha,k)=\frac{k^2}{4\pi}[\beta,[\alpha, \mathcal{E}]]c_D a^3.
\ee
Let us summarize what we have proved.
\begin{thm} \label{thm2.4}
    If $ka \ll 1$, then formula \eqref{eq2.15}, \eqref{eq2.25}, and \eqref{eq2.27} hold.
\end{thm}
Theorem \ref{thm2.3} and \ref{thm2.4} are our basic results for EM wave scattering by one small perfectly conducting body of an arbitrary shape.

\section{Many-body wave scattering}  \label{sec3}
Consider now the case of many-body EM wave scattering.

Let $D_m$ be small perfectly conducting body, $x_m \in D_m$ be an arbitrary point, $D:=\bigcup_{m=1}^M D_m, D':= \RRR \setminus D, D \subset \Omega \subset \RRR$.

Assume that formula \eqref{eq1.1} gives the distribution of small bodies in $\Omega$. The total number $M=M(a)$ of small particles, distributed in $\Omega$ is given by formula \eqref{eq1.2}.

The distance $d$ between closest neighboring particles is assumed to be large compared with $a$
\be \label{eq3.1}
    d \gg a
\ee
Let $L$ be a side of a cube $\Omega$ where small particles are distributed. Then $\left(\frac{L}{d}\right)^3=O(M)=O\left(\frac{1}{a^3}\right)$. Thus
\be \label{eq3.2}
    d=O(La).
\ee
Therefore, condition \eqref{eq3.1} holds if $L$ is large. If $L$ is fixed, then \eqref{eq3.1} holds if $N(x) \ll 1$, $d=O\left(\left(\int_\Omega N(x)dx\right)^{-1/3}a\right)$.

The many-body scattering problem consists of solving equation \eqref{eq1.4}-\eqref{eq1.7}, where now $D=\bigcup_{m=1}^M D_m$.

Let us look for the solution of the form
\begin{align}
    E(x)&=E_0(x)+\sum_{m=1}^M \nabla\times\int_{S_m}g(x,t)J_m(t)dt \nonumber\\
        &=E_0(x)+\sum_{m=1}^M [\nabla g(x,x_m),Q_m]+ \sum_{m=1}^M \nabla\times\int_{S_m}(g(x,t)-g(x,x_m))J_m(t)dt, \label{eq3.3}
\end{align}
where
\be \label{eq3.4}
    Q_m:=\int_{S_m} J_m(t)dt.
\ee
\begin{lem} \label{lem3.1}
    If
    \be \label{eq3.5}
        ka+\frac{a}{d} \ll 1,
    \ee
    then
    \be \label{eq3.6}
        J_m:=\left|\nabla\times\int_{S_m}(g(x,t)-g(x,x_m))J(t)dt\right|\ll \left|[\nabla g(x,x_m),Q_m]\right|:=I_m.
    \ee
\end{lem}
\begin{proof}
    One has
    \be \label{eq3.7}
        \left|\nabla g(x,x_m)\right| \le c\left(\frac{k}{d}+\frac{1}{d^2}\right),
    \ee
    where $c>0$ is a constant and $d=|x-x_m|$.

    Similarly,
    \be \label{eq3.8}
        \left|\nabla\times(g(x,t)-g(x,x_m))\right| \le ca\left(\frac{k^2}{d}+\frac{k}{d^2}+\frac{1}{d^3}\right),
    \ee
    where $|t-x_m|\le a, |x-x_m|=d$.
    Therefore,
    \be \label{eq3.9}
        J_m \le cQ\left(\frac{ak^2}{d}+\frac{ka}{d^2}+\frac{a}{d^3}\right), \quad I_m=QO\left(\frac{k}{d}+\frac{1}{d^2}\right).
    \ee
    Taking into account that $kd \ll 1$, one gets
    \be \label{eq3.10}
        \frac{J_m}{I_m} \le c\frac{ak^2d^2+kad+a}{d(1+kd)} \le c\left(ka+\frac{a}{d}\right) \ll 1.
    \ee
    Lemma \ref{lem3.1} is proved.
\end{proof}
From lemma \ref{lem3.1} it follows that {\em one can solve the many-body EM wave scattering problem by finding quantities $Q_m, 1 \le m \le M$, rather than the unknown functions $J_m(t)$. This allows one to solve numerically the scattering problem with very large $M$, provided that assumption \eqref{eq3.5} holds.}

 The solution is given by the formula
\be \label{eq3.11}
    E(x) \sim E_0(x)+\sum_{m=1}^M [\nabla_x g(x,x_m),Q_m].
\ee
Let us introduce the notion of the effective field: it is the field acting on a particular small particle $D_j$ from all other particles and from the incident field:
\be \label{eq3.12}
    E_e(x):=E_0(x)+\sum_{m \neq j} \nabla\times\int_{S_m}g(x,t)J_m(t)dt.
\ee
In the limit $a\to 0$ the effective field differs negligibly from the full field already at the distances of the order of $a$ because the radiation from a small particle is $O(a^3)$, as we proved in section \ref{sec2}.

If condition \eqref{eq3.5} holds, then the effective field, scattered by the j-th particle, can be calculated by the formula analogous to formula \eqref{eq2.24}:
\be \label{eq3.13}
    Q_j=- (\nabla\times E_e)(x_j)c_D a^3, \quad 1 \le j \le M.
\ee
Let us assume  for simplicity, that $c_D:=c_0$ does not depend on $j$. This, for example, happens if all the particles are of the same geometry. Then $Q_j$ is known if the quantities $A_j:=(\nabla\times E_e)(x_j)$ are known, $1 \le j \le M$.

Let us derive a linear algebraic system (LAS) for finding these quantities. This will give a numerical method for solving many-body EM wave scattering problem. Take curl of equation \eqref{eq3.12} and then put $x=x_j$ in the resulting equation. This yields
\be \label{eq3.14}
    A_j=A_{0j}- \sum_{m \neq j} \nabla_x\times[\nabla_x g(x,x_m),A_m(t)]|_{x=x_j}c_0 a^3,  \quad 1 \le j \le M,
\ee
where
\be \label{eq3.15}
    A_{0j}:=(\nabla\times E_0)(x_j).
\ee
If $A_j$ are found from \eqref{eq3.14}, then the solution to the scattering problem is given by formula \eqref{eq3.11}:
\be \label{eq3.16}
    E(x)=E_0(x)-\sum_{m=1}^M [\nabla_x g(x,x_m), A_m]c_0 a^3.
\ee
If $M$ is very large, then the order of LAS \eqref{eq3.14} is very large. Let us reduce this order drastically.

Consider a partition of the domain $\Omega$ into a union of cubes $\Delta_p, 1 \le p \le P, P \ll M$. Let $x_p$ be the center of $\Delta_p$ and $b$ be the side of $\Delta_p$. Assume that $b=b(a)$,
\be \label{eq3.17}
    b \gg d \gg a, \quad \lim_{a\to 0} b(a)=0,
\ee
where $d=d(a)$ is the distance between closest neighboring particles. Then there will be many particles in each of the cubes $\Delta_p$. Let us transform formula \eqref{eq3.16} as follows
\begin{align}
    E_q&:=E(x_q)=E_0(x_q)-\sum_{p=1}^P [\nabla_x g(x,x_p), A_p]|_{x=x_q} c_0a^3 \sum_{x_m \in \Delta_p}1 \nonumber \\
       &:= E_{0q}-c_0\nabla_x\sum_{p=1}^P g(x,x_p)A_p N(x_p)|\Delta_p|, \label{eq3.18}
\end{align}
where by formula \eqref{eq1.1} one has
\be \label{eq3.19}
    a^3 \sum_{x_m \in \Delta_p}1=N(x_p)|\Delta_p|(1+o(1)).
\ee
Here $|\Delta_p|$ is he volume of  the cube $\Delta_p$ and we have used the relations
\be \label{eq3.20}
    \nabla_x g(x,x_m) \sim \nabla_x g(x,x_p), \quad A_m \sim A_p, \quad \forall m: x_m \in \Delta_p.
\ee
These relations hold because the side of $\Delta_p$ tends to zero as $a\to 0$, and the functions $\nabla_x g(x,y), A(y)$ are continuous functions of $y \in \Delta_p$ when $x \not\in \Delta_p$. A similar transformation of LAS \eqref{eq3.14} yields
\be \label{eq3.21}
    A_q = A_{0q}-c_0\sum_{p\neq q}^P \nabla_x\times[\nabla_x g(x,x_p), A_p]|_{x=x_q}, \quad 1 \le q \le P.
\ee
The order of LAS \eqref{eq3.21} is $P \ll M$.

Equation \eqref{eq3.18} is a Riemannian sum for equation
\be \label{eq3.22}
    E(x) = E_0(x)-c_0\nabla\times\int_\Omega g(x,y)\nabla\times E(y) N(y) dy.
\ee
This equation describes the limiting field in the domain $\Omega$ when the number of particles tends to infinity while the size of the particles tends to zero, the particles are distributed by the law \eqref{eq1.1} and condition \eqref{eq3.5} holds.

Let us summarize the result.
\begin{thm} \label{thm3.2}
    Assume that condition \eqref{eq3.5} and equation \eqref{eq1.1} hold. The the many-body EM wave scattering problem \eqref{eq1.4}--\eqref{eq1.7} with $D=\bigcup_{m=1}^M D_m$ has a unique solution that can be computed by formula \eqref{eq3.11} with $Q_m$ given by formula \eqref{eq3.13} and the quantities $(\nabla\times E_e)(x_m):=A_m$ can be computed by solving LAS \eqref{eq3.14} or \eqref{eq3.21}. The effective field in $\Omega$ has a limit as $a \to 0$ and this limit solves equation \eqref{eq3.22}.
    Equation \eqref{eq3.22} has a solution in $C^2(\Omega)$ and this solution is unique.
\end{thm}
\begin{proof}
    Only the last statement of theorem \ref{thm3.2} is not yet proved.

    To prove it, assume that $E_0=0$, apply the operator $\nabla\times\nabla\times$ to the homogeneous version of equation \eqref{eq3.22} and use the formulas $$\nabla\times\nabla\times\nabla\times=(\nabla\cdot\nabla-\nabla^2)\nabla\times,\quad \nabla\cdot\nabla\times=0,$$
    and
    $$\nabla\cdot E=0, \quad -\nabla^2 g(x,y)=k^2 g(x,y)+\delta(x-y).$$ We obtain the following equation:
    $$\nabla\times\nabla\times E=k^2 E-c_0 N(x)\nabla\times\nabla\times E,$$ or
    $$\nabla\times\nabla\times E=\frac{k^2}{1+c_0 N(x)}E.$$ This equation and the relation $\nabla\cdot E=0$ imply
    \be \label{eq3.23}
        -\nabla^2 E-k^2 E+\frac{k^2 c_0 N(x)}{1+c_0 N(x)}E=0.
    \ee
    The field $E$ satisfies the radiation condition and equation \eqref{eq3.23} is a Schr\"{o}dinger equation with non-negative compactly supported potential
    $$q(x):=\frac{k^2 c_0 N(x)}{1+c_0 N(x)}.$$ Therefore, the solution $E(x)$ to equation \eqref{eq3.23} is zero.

     Let us prove this. Let $u$ be any Cartesian coordinate of $E(x)$. Then
    \begin{align}
        &[\nabla^2+k^2-q(x)]u=0 \quad \text{in } \RRR \label{eq3.24}\\
        &u_r-iku=o\left(\frac{1}{r}\right), \quad r \to \infty. \label{eq3.25}
    \end{align}
    Multiply equation \eqref{eq3.24} by $\overline{u}$, subtract complex conjugate equation \eqref{eq3.24} multiplied by $u$, and integrate over $B_R:=\{x: |x|  \le R\}$. The result is
    \be \label{eq3.26}
        0=\int_{B_R}(\overline{u}\nabla u-u\nabla\overline{u})dx=\int_{|x|=R}(\overline{u}u_r-u\overline{u}_r)ds= 2ik\int_{|x|=R} |u|^2ds+o(1).
    \ee
    Thus
    \be \label{eq3.27}
        \lim_{R\to \infty}\int_{|x|=R} |u|^2ds=0.
    \ee
    Since $q(x)=0$ outside $\Omega$, equation \eqref{eq3.24} yields
    \be \label{eq3.28}
        (\nabla^2+k^2)u=0, \quad |x|> R_0,
    \ee
    where $R_0$ is any number such that $B_{R_0} \supset \Omega$.

    By lemma 1 from \cite{R190}, p. 25, it follows from \eqref{eq3.28} and \eqref{eq3.27} that $u=0$ in $\RRR \setminus B_{R_0}$. By the unique continuation property for solutions to elliptic equation \eqref{eq3.24}, it follows that $u=0$ in $\RRR$. Thus, $E(x)=0$ in $\RRR$.

    This proves the last statement of Theorem \ref{thm3.2}.
\end{proof}

\section{Applications to materials science} \label{sec4}
Let us show that the new medium in $\Omega$ has a new refraction coefficient.

Apply the operator $\nabla\times\nabla\times$ to equation \eqref{eq3.22}, take it into account that
\be \label{eq4.1}
    \nabla\times\nabla\times E_0=k^2 E_0, \qquad \nabla\times\nabla\times\nabla\times=(\nabla\cdot\nabla-\nabla^2)\nabla\times,
\ee
and  $-\nabla^2 g=k^2 g+\delta(x-y)$. Thus,
\be \label{eq4.2}
    -\nabla^2\nabla\times \int_\Omega g(x,y) \nabla\times E(y) N(y)dy=\nabla\times \int_\Omega (k^2 g+\delta(x-y))\nabla\times E(y)N(y)dy.
\ee
Consequently, the equation resulting from \eqref{eq3.22} is:
\be \label{eq4.3}
    \nabla\times\nabla\times E =k^2 E-c_0 \nabla\times( \nabla\times E(x)N(x))=k^2 E-c_0 N(x) \nabla\times \nabla\times E-c_0[ \nabla N(x), \nabla\times E].
\ee
This implies
\be \label{eq4.4}
    \nabla\times\nabla\times E = \frac{k^2}{1+c_0 N(x)}-\frac{c_0}{1+c_0 N(x)}[\nabla N(x), \nabla\times E].
\ee
Let us interpret this equation physically. First, the new refraction coefficient is
\be \label{eq4.5}
    n^2(x)=\frac{k^2}{1+c_0 N(x)}.
\ee
By the refraction coefficient $n^2(x)$ one means the coefficient in the equation
\be \label{eq4.6}
    \nabla\times\nabla\times E = k^2n^2(x)E.
\ee
Secondly, to interpret the last term in \eqref{eq4.4}, consider equation \eqref{eq1.4} and assume that $\mu=\mu(x)$. Then
\be \label{eq4.7}
    \nabla\times\nabla\times E=i\omega\mu\nabla\times H+i\omega[\nabla\mu, H],
\ee
where $H=\frac{\nabla\times E}{i\omega\mu}$ and $ \nabla\times H=-i\omega\epsilon E$.

Thus,
\be \label{eq4.8}
    \nabla\times\nabla\times E=K^2(x)E+\left[\frac{\nabla\mu}{\mu},\nabla\times E\right],
\ee
where
\be \label{eq4.9}
    K^2(x)E=\omega^2 \epsilon \mu(x)=\frac{k^2}{1+c_0 N(x)}.
\ee
Consequently,
\be \label{eq4.10}
    \mu(x)=\frac{\mu}{1+c_0 N(x)},
\ee
and
\be \label{eq4.11}
    \frac{\nabla\mu}{\mu}=-\frac{c_0 \nabla N(x)}{1+c_0 N(x)}.
\ee
Therefore, equation \eqref{eq4.4} is exactly equation \eqref{eq4.8} with $\mu(x)$ defined in \eqref{eq4.10}.

We have proved the following theorem.
\begin{thm} \label{thm4.1}
    The new medium in $\Omega$ corresponds to the material with the refraction coefficient \eqref{eq4.5} and permeability \eqref{eq4.10}.
\end{thm}

\section{Conclusions} \label{sec5}
Basic results of this paper are formulated in Theorems \ref{thm2.3}, \ref{thm2.4}, \ref{thm3.2} and \ref{thm4.1}.


\begin{thebibliography}{99}

\bibitem{BH} C. Bohren, D. Huffman, Absorption and scattering of light by small particles, Wiley, New York, 1998.

\bibitem{Mie} G. Mie, Beitr\"{a}ge zur Optik Tr\"{u}ben Medien, spezielle kolloidaler Metall\"{o}sungen, Ann d. Phys., 25, (1908), 377-445.

\bibitem{R144} A. G. Ramm, {\em Iterative methods for calculating static fields and wave scattering by small bodies}, Springer Verlag, New York, 1982.

\bibitem{R190} A. G. Ramm, {\em Scattering by obstacles}, D.Reidel, Dordrecht, 1986, pp.1-442.

\bibitem{R476} A. G. Ramm, {\em Wave scattering by small bodies of arbitrary shapes}, World Sci. Publishers, Singapore, 2005.

\bibitem{R635} Ramm, A.G., 2013 {\em Scattering of Acoustic and Electromagnetic Waves by Small
Bodies of Arbitrary Shapes. Applications to Creating New Engineered Materials,} Momentum Press, New York.

\bibitem{Ra} Lord Rayleigh (J.W.Strutt), Scientific papers, Cambridge, 1922.

\end{thebibliography}
\end{document}